\documentclass[]{article}

\usepackage{PRIMEarxiv}

\usepackage[utf8]{inputenc} % allow utf-8 input
\usepackage[T1]{fontenc}    % use 8-bit T1 fonts
\usepackage{hyperref}       % hyperlinks
\usepackage{url}            % simple URL typesetting
\usepackage{booktabs}       % professional-quality tables
\usepackage{amsfonts}       % blackboard math symbols
\usepackage{nicefrac}       % compact symbols for 1/2, etc.
\usepackage{microtype}      % microtypography
\usepackage{lipsum}
\usepackage{fancyhdr}       % header
\usepackage{graphicx}       % graphics
\graphicspath{{media/}}     % organize your images and other figures under media/ folder
\usepackage{amsmath,amssymb,amsfonts,algpseudocode,algorithm,breqn,graphicx,caption,subcaption,multicol,color,amsthm}
\newtheorem{theorem}{Theorem}
\newtheorem{lemma}[theorem]{Lemma}

\title{A simple and efficient preprocessing step for convex hull problem 
}

\author{
  Mohammad Heydari\\
  Department of Computer Science, Khansar Campus \\
  University of Isfahan \\
  Isfahan, Iran\\   
  \texttt{m.heydari@khc.ui.ac.ir}\\
   \And
  Ashkan Khalifeh \\
  Department of Statistics \\
  Yazd University \\
  Yazd, Iran\\
  \texttt{khalifeh68@yahoo.com} \\
}

\begin{document}
\maketitle

\begin{abstract}
The present paper is concerned with a recursive algorithm as a preprocessing step to find the convex hull of $n$ random points uniformly distributed in the plane. For such a set of points, it is shown that eliminating all but $O(\log n)$ of points can derive the same convex hull as the input set. Finally it will be shown that the running time of the algorithm is $O(n)$.
\end{abstract}

% keywords can be removed

\section{Introduction}
For a given set of points $P$ in $\mathbb{R}^d$, the convex hull of $P$ is the smallest convex set containing $P$. Computing convex hull is a fundamental problem in many fields such as biology \cite{biology-app}, image processing \cite{image-processing,image-processing-2,graphics} and pattern recognition \cite{pattern-recognition2}. The problem of finding the convex hull of a set of $n$ points in $\mathbb{R}^2$ has extensively been studied and numerous elegant algorithms have been presented to address this problem in the literature. Among these algorithms, one can refer to Graham's Scan \cite{graham-scan} and Quickhull \cite{quickhull} that run in $O(n \log n)$ and Giftwrapping \cite{gift-wrapping} that runs in $O(nh)$, where $h$ is the number of points on the convex hull. An approach to address this problem includes a preprocessing step which involves finding and eliminating some points that are inside the convex hull. Then the rest of the points are fed into any convex hull algorithm. The Quickhull algorithm follows this approach. In addition, Akl and Toussaint \cite{toussaint} followed the point elimination approach to find the convex hull in worst-case running time $O(n \log n)$ and expected running time $O(n)$. However, no theoretical result in their work represents the order of eliminated points. An et al. \cite{modified-graham} and An \cite{optimization} presented algorithms that use the point elimination approach.

The convex hull of random points has a long history dating back to 1864 by Sylvester. The case in which the points are distributed by uniform distribution may be found in Bentley et al. \cite{bentley} and Golin and Sedgewick \cite{golin}, who specifically concentrated on the preprocessing step. The authors in \cite{bentley} and \cite{golin} concentrated on the point elimination approach for points that are uniformly distributed. Their preprocessing steps identify and eliminate all but $O(\sqrt{n})$ of the points. The results of \cite{bentley} are valid when there are more than 2,000,000 points and so the results are of theoretical interest. Golin and Sedgewick \cite{golin} studied the problem in unit square and eliminated $O(\sqrt{n})$ of points in $O(n)$ time. 

In the present paper, a new recursive algorithm will be considered as the point elimination preprocessing step for a set of $n$ points that are uniformly distributed in $\mathbb{R}^2$. The algorithm identifies and eliminates all but $O(\log n)$ of the points in $O(n)$ time and similar to previous algorithms, we may feed the remaining points to any convex hull algorithm. Most convex hull routines will run much more efficiently on few points than on many \cite{golin}, hence reducing the size of the input set is of practical interest. In addition, since the convex hull of the input set and the obtained $O(\log n)$ of points are identical, the preprocessing algorithm estimates an upper bound on the expected number of points located on the convex hull of the input set.

The rest of the paper is organized as follows: in Section 2, the preliminary lemmas are provided. Section 3 describes the preprocessing algorithm. Section 4 presents an analysis of the results and proves the correctness of the algorithm. Finally, Section 5 provides a concluding statement.

\section{Preliminaries}
This section provides the necessary definitions and lemmas used throughout the paper.

\textbf{Definition:} An extreme point among a set of points is a point whose $x$- or $y$-coordinates is minimum or maximum over the set, i.e., the topmost, leftmost, bottommost, and rightmost points.

In the following lemma, the bounding box of a set of uniformly distributed points in the plane is considered, and we investigate the distribution of points inside it. Formally speaking, a four-sided polygon formed by connecting a randomly selected point on each edge of the bounding box to the point on its adjacent edges is a quadrilateral. The formed right triangles inside the bounding box are called corners; see Figure \ref{fig:rec1}.

\begin{lemma} \label{bounding-box-n8}
	Consider a set of uniformly distributed points and its bounding box in the plane, as shown in  Figure \ref{fig:rec1}. Connecting the point on an edge of the bounding box to a point on its adjacent edges, will form a quadrilateral inside the bounding box. The expected number of points inside the quadrilateral and at each corner are $\dfrac{n}{2}$ and $\dfrac{n}{8}$, respectively.
	
	\begin{figure}
		\centering
		\includegraphics[width=0.8\linewidth]{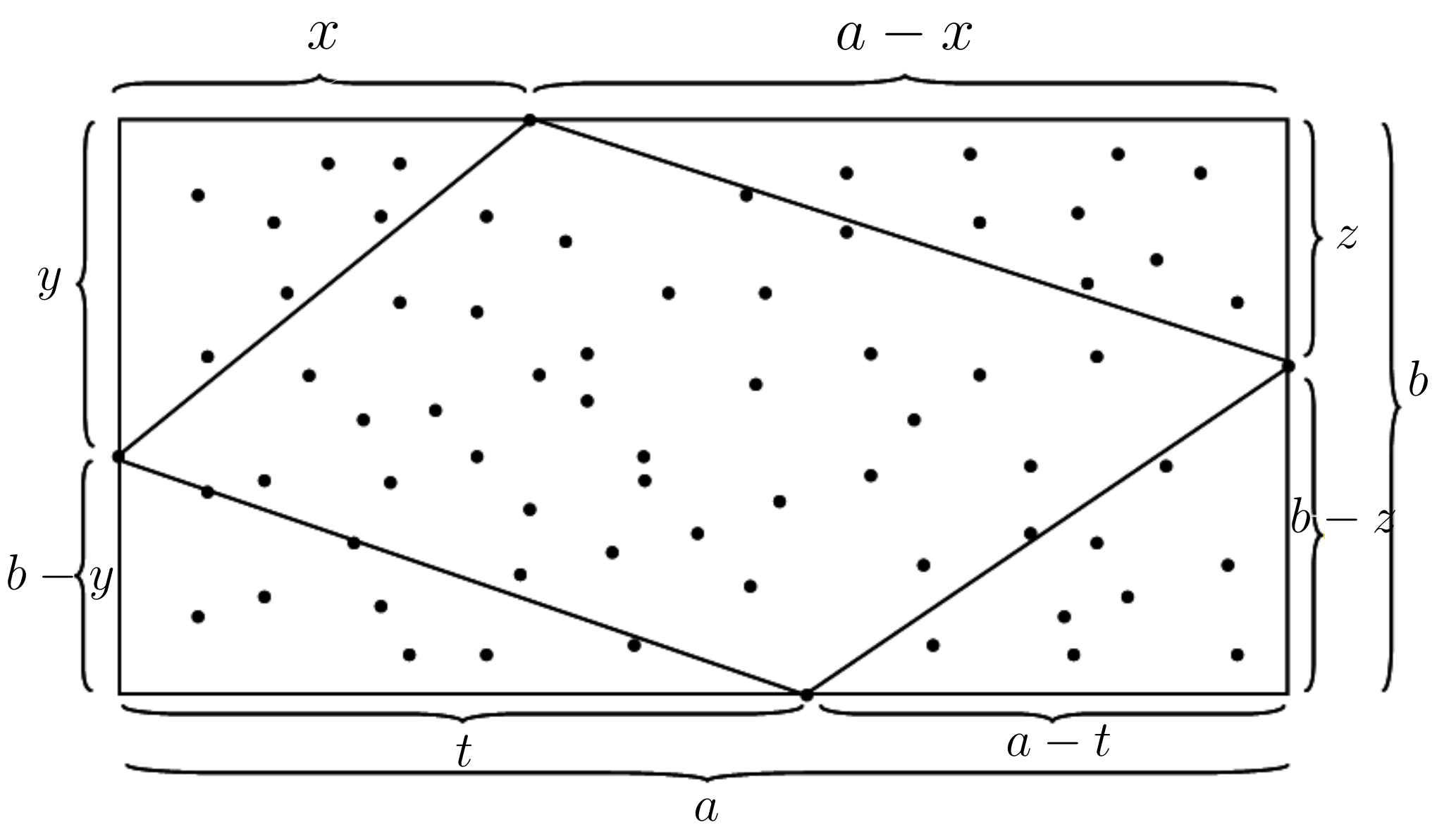}
		\caption{Illustration of quadrilateral made by connecting extreme points.}
		\label{fig:rec1}
	\end{figure}
\end{lemma}

\begin{proof}
	One point is randomly selected on each edge of the bounding box. Let us define $X, T, Y,$ and $Z$ with observed values $x, t, y, $ and $z$, respectively, as the random variables for the selected points. Therefore, $X$ and $T$ follow uniform distribution on the interval $[0, a]$ denoted by $U(0 ,a)$, and $Y$ and $Z$ follow  $U(0,b)$; see Figure \ref{bounding-box-n8}. Given the fact that each vertex of the quadrilateral is a random variable, the area of the quadrilateral is a random variable and is defined as follows
	$$
	N = ab-\dfrac{1}{2}[XY+(a-X)Z+(b-Y)T+(a-T)(b-Z) \big].
	$$
Hence, the mathematical expectation of the area of the quadrilateral is given by

	\begin{dmath*}
		E(N) =  ab \bigg(1 - \dfrac{1}{2ab}\big[E(X)E(Y) + aE(Z)
		- E(X)E(Z) + bE(T) - E(Y)E(T) +ab - aE(Z) - bE(T) + E(T)E(Z)\big]\bigg).
	\end{dmath*}
Since the mathematical expectation of a uniform random variable on an arbitrary interval $[c_1,c_2]$ is $\dfrac{c_1+c_2}{2}$, $E(X)$ and $E(T)$ are $\dfrac{a}{2}$. Similarly, $E(Y)$ and $E(Z)$ are $\dfrac{b}{2}$. Consequently, $E(N)$ is given by
	\begin{dmath*}
		E(N) =  ab \bigg(1 - \dfrac{1}{2ab}\big[\dfrac{ab}{4} + \dfrac{ab}{2}-\dfrac{ab}{4}+\dfrac{ab}{2}-\dfrac{ab}{4}+ab-\dfrac{ab}{2}-\dfrac{ab}{2}+\dfrac{ab}{4}\big] \bigg) = ab \bigg(1 - \dfrac{ab}{2ab} \bigg) = \dfrac{1}{2}ab.
	\end{dmath*}
	Therefore, the expected value of the area of $N$ is $\dfrac{ab}{2},$ where $ab$ is the area of the bounding box. Since the expected value of the area of $N$ is half of the area of the bounding box, and the points are uniformly distributed, the expected value of the number of points in $N$ is half of the number of points in the bounding box. 
	%This argument is proven as follows. The number of points inside the quadrilateral is a binomial random variable, say, $L  \sim Binomial(n, P)$. The mathematical expectation of its probability parameter, i.e. $E(P)$ is $0.5$, which is indeed the ratio of the mathematical expectation of two areas $(E(P)=\dfrac{E(A)}{ab})$. So, the expected value of the number of points inside the quadrilateral is given by
%$$E(L) = E(E(L|P)) = E(nP) = nE(P) = \dfrac{n}{2}.$$
It remains to prove the expected value of the area of each corner. Let $Q$ be the random variable of area of the top left corner. Then we see that
	$$
	Q = \dfrac{1}{2}XY.
	$$
Given the fact that $X$ and $Y$ are independent random variables, one can easily show that
	$$
	E(Q) = \dfrac{1}{2} E(X)E(Y) = \dfrac{1}{2}\times \dfrac{a}{2}\times\dfrac{b}{2} = \dfrac{ab}{8}.
	$$
	By the same argument, it is easy to see that the expected value of the number of points for each corner is $\dfrac{n}{8}$ and it completes the proof.
\end{proof}

In the following lemma, we investigate the distribution of points in a right triangle. We assumed there is at least one point on each leg because otherwise, we can find the extreme points and draw new legs that pass through the extreme points; see Figure \ref{fig:tr2}.

\begin{lemma} \label{right-triangle-distribution}
	Suppose we are given a right triangle with $n \geq 4$ uniformly distributed points inside it. Let there be at least one point on each of its legs. If we randomly take one point on each leg and connect them, then the expected value of the number of points inside the inner triangle is $\dfrac{n}{4}$.
\end{lemma}
\begin{proof}
	Consider Figure \ref{fig:triangle}. Let $X_1$ and $ Y_1$ be the following uniform random variables
	\begin{align*}
	X_1 & \sim U(O,a),\\
	Y_1 & \sim U(O,b)
	\end{align*}
	for the selected points. Let us denote the area of a triangle $ABC$ as $S(ABC)$ and the mathematical expectation of $S(OX_1Y_1)$ as $H$. Then one can easily show that
	\begin{equation*}
		E(H)= E(\dfrac{1}{2}Y_1X_1) = \dfrac{1}{2}E(Y_1)E(X_1) = \dfrac{1}{2} \times \dfrac{b}{2} \times \dfrac{a}{2} = \dfrac{1}{8}ab = \dfrac{1}{4}S(Oab).
	\end{equation*}
	Since the expected value of the area of triangle $OY_1X_1$ is a quarter of the area of triangle $Oab$, the expected value of the number of points in triangle $OY_1X_1$ is $\dfrac{n}{4}$.
\end{proof}
	
	\begin{figure}
		\centering
		\begin{subfigure}{.5\textwidth}
			\centering
			\includegraphics[width=.5\linewidth]{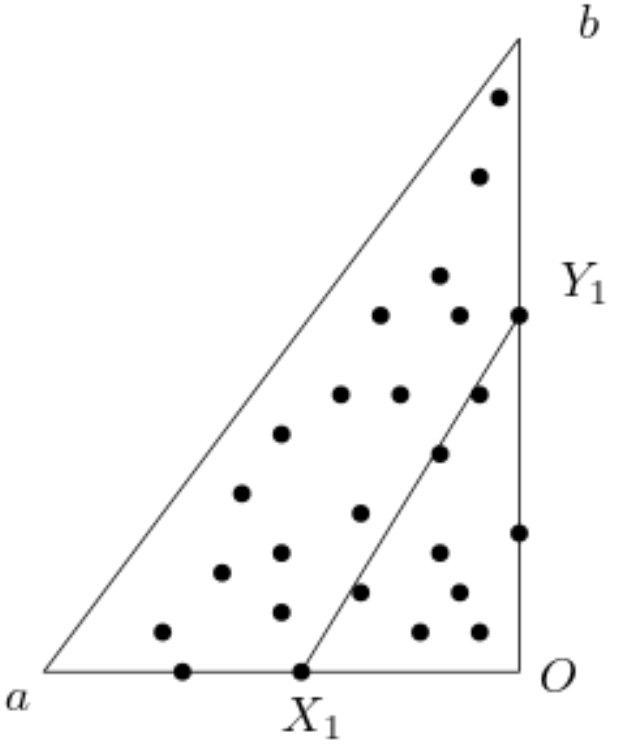}
			\caption{Distribution of points in triangle.}
			\label{fig:triangle}
		\end{subfigure}%
		\begin{subfigure}{.57\textwidth}
			\centering
			\includegraphics[width=.4\linewidth]{2-b.pdf}
			\caption{Illustration of right triangle made by new legs.}
			\label{fig:tr2}
		\end{subfigure}
		\caption{Right triangles.}
	\end{figure}

In the next section, the proposed algorithm is described and then its correctness and running time are proven.

\section{Algorithm}
In this section, the preprocessing algorithm is discussed. The algorithm takes a set $S$ of points as input and returns a set $A$ of points such that the covexhull of set $S$ and set $A$ will be identical. Throughout the paper, for a given set of points, the obtained extreme points are named in clockwise order as $p_t,p_r,p_b,p_l$ such that $p_t$ has the maximum value for the $y$-coordinate. In case there is more than one point whose $x$- or $y$-coordinate is minimum or maximum, we take one of them.

The proposed algorithm on a subset of points recursively selects and stores some extreme points in $A$. The extreme points in $i$th recursion are denoted by $p_t^{(i)},p_r^{(i)},p_b^{(i)},p_l^{(i)}$. As shown in Figure \ref{fig:algorithm}, the algorithm starts with finding the extreme points and stores $p_t^{(1)},p_r^{(1)},p_b^{(1)},p_l^{(1)}$ in the set $A$. Then recursively for the top right corner, it finds the extreme points and stores the $p_t^{(2)},p_r^{(2)}$ in $A$. For the bottom right, bottom left and top left corners it stores $p_r^{(2)},p_b^{(2)}$ and $p_b^{(2)},p_l^{(2)}$ and $p_l^{(2)},p_t^{(2)}$ in $A$, respectively. The algorithm repeats this procedure for the points above the segment lines $p_t^{(2)},p_r^{(2)}$ in the top right corner, and $p_t^{(2)},p_l^{(2)}$ in the top left corner, and the points below the segment lines $p_r^{(2)},p_b^{(2)}$ in the bottom right corner and $p_b^{(2)},p_l^{(2)}$ in the bottom left corner, until there are at most two points remaining.

The pseudo-code of the algorithm described above is given in Algorithm \ref{rh}. The strings TR, TL, BR and BL represent the top right, top left, bottom right, and bottom left right triangles of Figure \ref{fig:rec1}, respectively.

\begin{figure}
	\centering
	\includegraphics[width=0.6\linewidth]{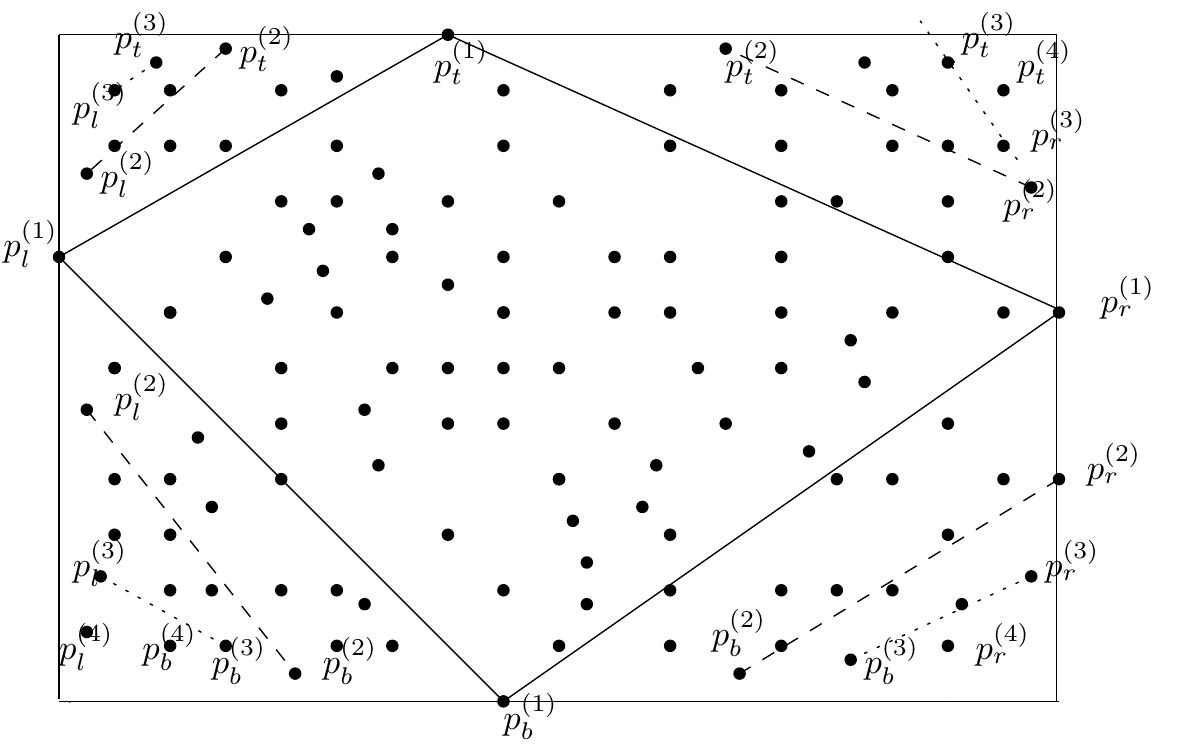}
	\caption{Illustration of execution of the algorithm.}
	\label{fig:algorithm}
\end{figure}

\begin{algorithm}\label{alg1}
	\caption{Preprocessing Algorithm}\label{rh}
	\textbf{Input: } Set $S$ of $n$ points\\
	\textbf{Output: } Set $A$ of points
	\begin{multicols}{2}
	\begin{algorithmic}[1]
		\State \textbf{IF} $|S| <= 3$ \textbf{ THEN }		
		\State \hspace{10pt} \textbf{ RETURN } $S$
		\State Find extreme points $p_t,p_r,p_b,p_l$.
		\State Let
		\begin{align*}
		S_1 &: \text{points above line segment } p_t p_r \\
	 	S_2 &: \text{points below line segment } p_r p_b\\		
		S_3 &: \text{points below line segment } p_b p_l\\ 
		S_4 &: \text{points above line segment } p_l p_t		
		\end{align*}
		\State $A = \{p_t,p_r,p_b,p_l\}$
		\State PointSelection($S_1,"TR"$)
		\State PointSelection($S_2,"BR"$)
		\State PointSelection($S_3,"BL"$)
		\State PointSelection($S_4,"TL"$)		
		\Procedure{PointSelection}{$S_r, corner$}
		\State \textbf{IF} $|S_r| <= 2$ \textbf{ THEN }
		\State \hspace{10pt} $A = A \cup S_r$
		\State \hspace{10pt} \textbf{ RETURN }	
		\State Find extreme points $p_t$,$p_r$,$p_b$,$p_l$.
		\State On point set $S_r$, let
		\begin{align*}
		S_1 &: \text{points above line segment } p_t p_r \\
	 	S_2 &: \text{points below line segment } p_r p_b\\		
		S_3 &: \text{points below line segment } p_b p_l\\ 
		S_4 &: \text{points above line segment } p_l p_t		
		\end{align*}					
		\State \textbf{IF} corner == "TR" \textbf{ THEN }
		\State \hspace{20pt} $A = A \cup \{p_t,p_r\}$
		\State \hspace{20pt} PointSelection($S_1,"TR"$)
		\State \textbf{IF} corner == "BR" \textbf{ THEN }
		\State \hspace{20pt} $A = A \cup \{p_r,p_b\}$
		\State \hspace{20pt} PointSelection($S_2,"BR"$)		
		\State \textbf{IF} corner == "BL" \textbf{ THEN }			
		\State \hspace{20pt} $A = A \cup  \{p_b,p_l\}$
		\State \hspace{20pt} PointSelection($S_3,"BL"$)		
		\State \textbf{IF} corner == "TL" \textbf{ THEN }		
		\State \hspace{20pt} $A = A \cup  \{p_l,p_t\}$								
		\State \hspace{20pt} PointSelection($S_4,"TL"$)		
		\EndProcedure

	\end{algorithmic}
	\end{multicols}
\end{algorithm}

\section{Correctness and Analysis}
In this section the correctness of the algorithm is proven in Theorem \ref{correctness}. The expected value of the number of points stored by the algorithm is estimated in Theorem \ref{estimate}. In Theorem \ref{running-time}, we will prove that the running time of the algorithm is $O(n)$.

\begin{theorem}\label{correctness}
	The set $A$ contains every point on the convex hull of the input set.
\end{theorem}
\begin{proof}
	In the first step, the algorithm stores extreme points of the input set. We prove the theorem for the top right corner; it can easily be generalized to other corners in a similarly.  At $i$th recursion, for $i \ge 2$, the points below the segment line $p_t^{(i)}p_r^{(i)}$ are inside the quadrilateral $p_t^{(i-1)}p_t^{(i)}p_r^{(i)}p_r^{(i-1)}$ and they are removed from further computations. Since they are inside a quadrilateral, they cannot be present on the convex hull of $S$. Therefore, set $A$ contains all points on the convex hull of $S$.
\end{proof}

In Theorem \ref{estimate}, we prove that the expected size of set $A$ is $O(\log n)$.  By Theorem \ref{correctness}, since set $A$ contains every point on the convex hull of the input set, it can be concluded retaining only $O(\log n)$ of points can derive the same convex hull as the input set.
\begin{theorem}\label{estimate}
	The expected size of set $A$ is $O(\log n)$.
\end{theorem}
\begin{proof}
	From Lemma \ref{bounding-box-n8}, the expected number of points at each corner is $m = n/8$, and half of the points are eliminated in the first step of the algorithm. Lemma \ref{right-triangle-distribution} also states that at each recurrence of the algorithm, the expected number of points at each corner to enter the next recursion is $O(m_r/4)$, where $m_r$ is the number of remaining points inside a corner at the $r$th recurrence. Hence, the expected number of iterations of the algorithm for a corner is $O(\log_4 {m})$. For each corner, the algorithm takes and stores $2$ points in each recursion. Therefore, an upper bound on the total number of stored points in $A$ is  
	\begin{equation*}
	O(4 \times 2\log_4 m) = O(\log n).
	\end{equation*}
\end{proof}

In the following theorem, we prove that the expected running time of the algorithm is linear.

\begin{theorem}\label{running-time}
	The expected running time of the algorithm is $O(n)$.
\end{theorem}
\begin{proof}
	The running time of the algorithm involves the running time of the first step that is $O(n)$ plus the running time of each recursion. Let $m = \lceil n/8 \rceil$. The expected iterations of the algorithm is $O(\log_4 m)$. Line 15 of the algorithm that determines the position of the points with respect to the segment lines takes $\Theta(m)$. Finding extreme points $p_t, p_r, p_b, p_l$ also takes $\Theta(m)$. Thus the running time of each recursion  is
	$$
	T(m) = T(m/4) + \Theta(m),
	$$
	and given the Master Theorem, the solution to this recursive equation is $O(m)$. Therefore, the total running time is:
	$$
	T(n) = O(n) + O(m) = O(n) + O(4 \times n/8) = O(n).
	$$
\end{proof}

\section{Conclusion}
This paper proposed an efficient yet simple recursive algorithm as a preprocessing step to the convex hull problem. We assumed the points were randomly distributed in the plane by the uniform distribution. The algorithm eliminates all but $O(\log n)$ of points in $O(n)$, and to the best of our knowledge, it is an improvement to all previous similar works which retain $O(\sqrt{n})$ points. Then, the remaining points are fed to any existing convex hull algorithm.

\end{document}